\documentclass[letterpaper, 10 pt, conference]{ieeeconf}  

\IEEEoverridecommandlockouts                              
\usepackage{lipsum}
\usepackage{amsmath} 
\usepackage{bbm}
\usepackage{amssymb}

\usepackage{amsthm}
\usepackage{mathtools}
\usepackage{todonotes}

\usepackage{todonotes}

\theoremstyle{definition}
\newtheorem{definition}{Definition}
\newtheorem{theorem}{Theorem}
\newtheorem{lemma}{Lemma}
\newtheorem{assumption}{Assumption}
\newtheorem{corollary}{Corollary}
\newtheorem{problem}{Problem}

\newtheorem{property}{Property}

\usepackage[linesnumbered,ruled]{algorithm2e} 

\usepackage{hyperref}
\hypersetup{
    colorlinks=true,
    linkcolor=black,
    filecolor=black,      
    urlcolor=blue,
    }
    
\usepackage{graphicx}
\graphicspath{ {resources/} }
\usepackage{booktabs}
\usepackage{float}
\usepackage{cite}

\let\labelindent\relax 
\usepackage{enumitem}

\title{\LARGE \bf
Beyond Reynolds: A Constraint-Driven Approach to Cluster Flocking
}

\author{Logan E. Beaver, \textit{Student Member, IEEE}, Andreas A. Malikopoulos, \textit{Senior Member, IEEE}
	\thanks{This research was supported in part by ARPAE's NEXTCAR program under the award number DE-AR0000796 and by the Delaware Energy Institute (DEI).}
	\thanks{The authors are with the Department of Mechanical Engineering at the University of Delaware in Newark, DE 19716, USA
		{ (emails: \texttt{lebeaver@udel.edu}; \texttt{andreas@udel.edu)}} }%
}

\begin{document}

\maketitle

\begin{abstract}
In this paper, we present an original set of flocking rules using an ecologically-inspired paradigm for control of multi-robot systems.
We translate these rules into a constraint-driven optimal control problem where the agents minimize energy consumption subject to safety and task constraints.
We prove several properties about the feasible space of the optimal control problem and show that velocity consensus is an optimal solution.
We also motivate the inclusion of slack variables in constraint-driven problems when the global state is only partially observable by each agent.
Finally, we analyze the case where the communication topology is fixed and connected, and prove that our proposed flocking rules achieve velocity consensus.
\end{abstract}

\section{Introduction}

Robotic swarm systems have attracted considerable attention in many applications, such as transportation \cite{Malikopoulos2018}, construction \cite{Lindsey2012ConstructionTeams}, and surveillance\cite{Corts2009}.
Controlling emergent flocking behavior has been of particular interest to several researchers since the seminal paper by Reynolds \cite{Reynolds1987}, which introduced three heuristic rules for flocking: flock centering, collision avoidance, and velocity matching.
Flocking has many practical applications, such as mobile sensing networks, coordinated delivery, reconnaissance, and surveillance  \cite{Olfati-Saber2006FlockingTheory}.
 
This paper explores the emergent phenomenon of flocking through a constraint-driven optimal control framework. Since we apply an ecologically inspired robotics approach \cite{Egerstedt2018RobotAutonomy}, it is important to distinguish between the two modes of flocking, as described by Bajec and Heppner \cite{Bajec2009OrganizedBirds}.
This work focuses on \emph{cluster} flocking, which describes the bulk motion of small birds, such as sparrows and pigeons.
This is in contrast to \emph{line} flocking, which describes the movement of geese and other large birds.
The hypothesis made by the ecology community is that sensor fusion is the main benefit of cluster flocking \cite{Thiollay1998FlockingHypothesis}. It has been proposed that cluster flocking allows individuals to estimate the total flock size in order to regulate breeding \cite{Bajec2009OrganizedBirds}. Cluster flocking could have value as a localization technique in some engineering applications.

One open question in the ecology literature \cite{Bajec2009OrganizedBirds} is whether a complex systems approach \cite{Malikopoulos2016b} is the proper way to model flocking.
In this paper, we devise a set of local flocking rules and translate them into a constraint-driven optimal control problem.
We then show several properties of our model and prove that the desired flocking behavior emerges in the case that the communication and sensing topology is fixed.
The optimal control problem is similar to previous work in constraint-driven flocking \cite{Beaver2020AnFlocking, Ibuki2020Optimization-BasedBodies}, where the authors implemented Reynolds flocking behavior. However, in this paper,  we propose an original set of flocking rules under the constraint-driven paradigm.
The optimal control policy, which emerges from our proposed rules, is related to existing constraint-driven optimization approaches using control barrier functions \cite{Notomista2019Constraint-DrivenSystems,Egerstedt2018RobotAutonomy}. In addition, we allow agents to plan over a horizon rather than just reacting to the environment and other agents.
Our approach also explicitly allows for the prioritization of safety constraints over the flocking behavior, which has only been recently explored in control barrier approaches \cite{Lindemann2019ControlTasks}.

The main contributions of this paper are as follows: (1) we present an original constraint-driven model for cluster flocking that is explicitly energy-minimizing and does not suffer from several problems present in potential field methods, (2) we discuss the challenges of solving constraint-driven problems over a planning horizon as opposed to reactive methods, and (3) we provide a proof of convergence to velocity consensus under a fixed network topology.

The remainder of this paper is organized as follows. In Section \ref{sec:flocking}, we discuss the flocking rules proposed by Reynolds and some of their shortcomings. In Section \ref{sec:problem}, we formulate the cluster flocking problem. In Section \ref{sec:solution}, we present several properties of our proposed model and provide guarantees for convergence to flocking for a special case. Finally, 
in Section \ref{sec:conclusion}, we provide concluding remarks and some directions for future work.

\section{Reynolds Flocking} \label{sec:flocking}

A significant amount of research on designing flocking behavior for boids (bird-oids) in multi-agent systems is based on the seminal work by Reynolds \cite{Reynolds1987}. To achieve flocking, Reynolds proposed three heuristic rules to govern the behavior of individual boids: 
(1) collision avoidance (avoid collisions with nearby flockmates),
(2) velocity matching (attempt to match velocity with nearby flockmates),  and 
(3) flock centering (attempt to stay close to nearby flockmates).
These rules have been the basis for most flocking controllers in the literature. Generally, these rules are imposed by coupling an artificial potential field with velocity consensus.A rigorous analysis on the convergence properties of potential-driven Reynolds flocking was presented in \cite{Tanner2007}.
Reynolds suggested two additional rules to guarantee safety and control of the flock \cite{Reynolds1987}: steer to avoid environmental hazards and move toward a ``migratory urge'' location specified by the designer.
Reynolds noted that using an artificial potential field for obstacle avoidance is not realistic, as it tends to push boids perpendicular to their direction of motion. 
Several extensions inherent to potential field problems are explored in \cite{Koren1991PotentialNavigation}, including local minima which trap agents and steady oscillations appearing in the agent trajectories. 
The design of an artificial potential field for the individual boids is still an open question. As discussed by V\'as\'arhelyi et al. \cite{Vasarhelyi2018OptimizedEnvironments}, the design of an optimal potential field is unintuitive and the effect of system parameters on the flock behavior can not be easily predicted. 

In a recent paper, the idea of constraint-driven (or ecologically-inspired) optimization to the control of multi-robot systems was presented \cite{Notomista2019Constraint-DrivenSystems}.
Under constraint-driven control approaches, agents seek to minimize energy consumption subject to a set of constraints. Thus, an agent's behavior can be understood from the constraints that become active during operation.
In the next section, we propose constraint-driven rules to achieve flocking and translate those rules into an optimal control problem.


\section{Problem Formulation} \label{sec:problem}

Consider a flock of  $N\in\mathbb{N}$ boids indexed by the set $\mathcal{A} = \{ 1, 2, \dots, N\}$. Each boid $i\in\mathcal{A}$ follows double integrator dynamics,
\begin{align}
	\dot{\mathbf{p}}_i(t) &= \mathbf{v}_i(t), \label{eq:pDynamics} \\
	\dot{\mathbf{v}}_i(t) &= \mathbf{u}_i(t), \label{eq:vDynamics}
\end{align}
where $t\in\mathbb{R}_{\geq0}$ is time, and $\mathbf{p}_i(t), \mathbf{v}_i(t), \mathbf{u}_i(t) \in \mathbb{R}^2$ are the position, velocity, and control input for boid $i$, respectively. Thus, the state of any boid $i\in\mathcal{A}$ is given by $\mathbf{x}_i(t) = [\mathbf{p}_i(t)  ~\mathbf{v}_i(t)]^T.$
The speed and control input for each boid are constrained such that
\begin{align}
    ||\mathbf{v}_i(t)|| \leq v_i^\text{max}, \label{eq:vConstraint}\\
    ||\mathbf{u}_i(t)|| \leq u_i^\text{max}. \label{eq:uConstraint}
\end{align}
	
For any pair of boids $i,j\in\mathcal{A}$, the relative displacement between them is described by the vector
\begin{equation} \label{eq:s}
	\mathbf{s}_{ij}(t) = \mathbf{p}_j(t) - \mathbf{p}_i(t).
\end{equation}

Each boid interacts with the flock through its neighborhood, defined next.

\begin{definition}[Neighborhood] \label{def:neighborhood}
	We define the \emph{neighborhood} of each boid $i\in\mathcal{A}$, denoted $\mathcal{N}_i(t)$, as the set of $k$-nearest neighbors of boid $i$ at time $t\in\mathbb{R}_{\geq0}$, where $k\in\mathbb{N}$ and $i\not\in\mathcal{N}_i(t)$.
\end{definition}
	
The neighborhood of each boid $i\in\mathcal{A}$ may switch over time, and  $i$ can communicate with, and sense, any boid $j\in\mathcal{N}_i(t)$.
Observations by Ballerini et al. \cite{Ballerini2008InteractionStudy} provide strong evidence that natural flocks are formed following a $k$-nearest neighbors rule. Cristiani, Frasca, and Piccoli \cite{Cristiani2011EffectsGroups} indicate that application of $k$-nearest neighbors ``$\ldots$ does not intend to imply that animals sensing capabilities extend to an unlimited range, but rather that group dynamics happen in a relatively small area.'' Following this rationale, the neighborhood, as defined in Definition \ref{def:neighborhood}, has several advantages: (1) it is distance agnostic, i.e., there can be no isolated boids that escape the flock, (2) it is a constant size, so the smallest connected group of boids is at least size $k$, and (3) it only considers a constant-size subset of nearby boids which implies that the information required by each boid is independent of the total flock size. 

Finally, we use the following model for the rate of energy consumption by any boid $i\in\mathcal{A}$,
\begin{equation} \label{eq:energyModel}
    \dot{e}_i(t) = ||\mathbf{u}_i(t)||^2,
\end{equation}
i.e., the rate of energy consumption is proportional to the $L^2$ norm of the control input. This is a reasonable model for our very general boid, as energy consumption is monotonic with the control input in many real systems.

Given the modeling framework above, we propose the following constraint-driven flocking rules: (1) \emph{energy minimization} to drive energy consumption to a minimum, (2) \emph{collision avoidance} to avoid collision with any flockmates, and (3) \emph{aggregation} to stay within a fixed distance of the average position of nearby flockmates.
Following these rules, the boids spend only the minimum energy necessary to ensure that safety and aggregation are achieved. It is expected that the combination of energy minimization and aggregation will eventually move all boids in the same direction as their neighborhood.
Energy-minimized obstacle avoidance can also be explicitly captured as a constraint in this framework.
The first objective, e.g., energy minimization, can be ensured by formulating each boid's behavior as an optimal control problem with the cost given by \eqref{eq:energyModel}. The second objective, e.g., collision avoidance, can be guaranteed between boids with the following safety constraint,
\begin{equation}\label{eq:safety}
	||\mathbf{s}_{ij}(t)||^2 \geq 4R^2, \quad\forall j\in\mathcal{A},\quad\forall\,t\in\mathbb{R}_{\geq 0},
\end{equation}
where $R\in\mathbb{R}_{>0}$ is the radius of each boid. The squared form in \eqref{eq:safety} is used to ensure that the derivative is smooth at $||\mathbf{s}_{ij}(t)|| = 2R$.
Additional constraints can be formulated to avoid static and dynamic obstacles in the environment \cite{Rezaee2014,VanDenBerg2011ReciprocalObstacles}.
Finally, we denote the center of boid $i\in\mathcal{A}$'s neighborhood as
\begin{equation} \label{eq:neighborhoodCtr}
	\mathbf{c}_i(t) = \frac{1}{k} \sum_{j\in\mathcal{N}_i(t)} \mathbf{p}_i(t),
\end{equation}
where $k$ corresponds to the $k$-nearest neighbors. We enforce aggregation through the \emph{flocking constraint},
\begin{equation} \label{eq:flocking}
    g_i(t) = ||\mathbf{p}_i(t) - \mathbf{c}_i(t)||^2 - D^2 \leq 0, \quad \forall t\in\mathbb{R}_{\geq0},
\end{equation}
where $D\in\mathbb{R}_{>0}$ is the system parameter corresponding to the flocking radius. Again, we use the equivalent squared form of the constraint to guarantee smoothness of the derivative.
In our decentralized optimal control problem, \eqref{eq:safety} is considered our \emph{safety constraint}, while \eqref{eq:flocking} is the \emph{task constraint}. 
In case where no feasible trajectory can be found, then we may relax the task constraint by applying least-penetrating control techniques \cite{Lindemann2019ControlTasks} or introducing slack variables \cite{Ibuki2020Optimization-BasedBodies}.
	
The standard solution methodology in the constraint-driven literature is to encapsulate the task and safety constraints within a control barrier function \cite{Notomista2019Constraint-DrivenSystems,Egerstedt2018RobotAutonomy,Ibuki2020Optimization-BasedBodies}. The boids will then apply gradient flow to drive energy consumption to a stationary point while satisfying the constraints \cite{Egerstedt2018RobotAutonomy}.

However, this approach is reactive and does not explicitly allow for planning or cooperation between the boids. In this paper, we seek to plan a trajectory over some finite-time horizon, $[t_i^0, t_i^f] \subset \mathbb{R}_{\geq0},$ for each boid $i\in\mathcal{A}$.
Next, we formulate the problem.
\begin{problem} \label{prb:flocking}
	For each boid $i\in\mathcal{A}$, consider the decentralized energy-minimization problem over the horizon $[t_i^0, t_i^f]\subset\mathbb{R}_{\geq 0}$,
	\begin{gather}
		\min_{\mathbf{u}_i(t)} \, \frac{1}{2} \int_{t_i^0}^{t_i^f} ||\mathbf{u}_i(t)||^2 \, + \alpha_i\, \eta_i^2(t) \, dt,\\\notag
	    \text{s.t.:}~ \mathbf{x}_i(t_i^0) = \mathbf{x}_i^0,
	    \eqref{eq:pDynamics}, \eqref{eq:vDynamics}, \eqref{eq:vConstraint}, \eqref{eq:uConstraint}, \eqref{eq:safety},
        \text{and}~ g_i(t) \leq \eta_i^2(t),
	\end{gather}
	where $\mathbf{x}_i^0$ is the initial state of boid $i$, $g_i(t)$ is the flocking constraint given by \eqref{eq:flocking}, $\eta_i^2(t)$ is a a slack variable that allows safety to be prioritized over flocking, and $\alpha_i$ is a system parameter to weight energy consumption versus the need to satisfy the flocking constraint.
\end{problem}

In contrast to Reynolds' flocking rules, in our  formulation in Problem \ref{prb:flocking}, we do not impose a desired inter-boid spacing. Thus our system will not converge to the $\alpha$-lattice formation, which is the optimal solution to minimizing a potential field \cite{Olfati-Saber2006FlockingTheory}. Additionally, our system does not require velocity alignment, and it instead emerges naturally as a solution to Problem \ref{prb:flocking}.
For the solution of Problem \ref{prb:flocking}, we impose the following assumptions.

\begin{assumption}\label{as:empty}
    There are no external disturbances or obstacles.
\end{assumption}

We impose Assumption \ref{as:empty} to evaluate the idealized performance of the proposed algorithm. It is well known that optimal control can be fragile with respect to noise and disturbances, and this assumption may be relaxed by applying robust optimal control.
    
\begin{assumption}\label{as:perfect}
    There are no errors or delays with respect to communication and sensing.
\end{assumption}
The strength of Assumption \ref{as:perfect} is application dependent. In general, it has been shown that sparse updates to re-plan trajectories are sufficient for this type of problem \cite{Hu2018Self-triggeredSystems}. However, these delays may become significant for fast-moving flocks in constrained environments.
    
\begin{assumption}\label{as:lowDensity}
    Each boid has a low-level onboard controller that can track the prescribed optimal trajectory.
\end{assumption}
Assumption \ref{as:lowDensity} may be strong for certain applications. This assumption may be relaxed by including kinematic constraints on the motion of each boid, or by considering more complicated dynamics in Problem \ref{prb:flocking}.

\subsection{The Information Problem}

The formulation of Problem \ref{prb:flocking} has two distinct issues that repeatedly occur in the decentralized control literature.
The first occurs when a boid only makes partial observations of the entire flock, i.e., $|\mathcal{N}_i(t)| < N - 1$ for $i \in\mathcal{A}$.
This results in a non-classical information structure, and results from centralized control do not generally apply \cite{Dave2019a}.
This problem may be circumvented by sharing information throughout the network \cite{Morgan2016}; however, this is prohibitively expensive for large flocks of boids.
This problem may also be solved by employing event-triggered control, where each boid re-solves  Problem \ref{prb:flocking} whenever it receives new information \cite{Beaver2019AGeneration}. This method only requires local information that is readily available to each boid.

The second issue is a potential simultaneous action of the boids. Consider two boids $i,j\in\mathcal{A}$ which seek a solution to Problem \ref{prb:flocking} at $t = t_i^0 = t_j^0$.
Boid $i$ requires the trajectory of boid $j$ to calculate $\mathbf{c}_i(t)$, while boid $j$ requires the trajectory of boid $i$ to calculate $\mathbf{c}_j(t)$.
This coupling guarantees that boids $i$ and $j$ can never satisfy the constraint \eqref{eq:flocking}, as its value is not known until after a trajectory has been generated.
The most straightforward solution to this problem is to impose some priority ordering on the boids, then to have them solve for their trajectories sequentially \cite{Beaver2019AGeneration,Turpin,Beaver2020DemonstrationCity,chalaki2019optimal}. Another approach is to allow the boids to make decisions asynchronously, which implicitly imposes an order \cite{Lin2004a} based on the boids' hardware specifications.

A standard approach to resolving the information problem in multi-agent systems is to apply decentralized model predictive control \cite{Luis2019TrajectoryControl}. In this case, each boid solves Problem \ref{prb:flocking}, the it follows the prescribed trajectory for some period of time, and then it re-solves Problem \ref{prb:flocking}. The trajectories  generated at the previous time step are used as estimates for the new trajectories generated by each neighbor.
This processes addresses both the simultaneous action and partial observation information problems.
The derivation of a model predictive controller for Problem \ref{prb:flocking} is beyond the scope of this paper. However, recent approaches by Zhan and Li \cite{Zhan2013FlockingMeasurements} and Lyu et al. \cite{Lyu2019MultivehicleControl} for constrained distributed model predictive control of flocking systems mat be adapted to Problem \ref{prb:flocking}.



\section{Optimal Solution Properties} \label{sec:solution}

Next, we present several properties of the system described by Problem \ref{prb:flocking}. 
First we examine the discontinuities of the system which are imposed by neighborhood switches through the following Lemmas.

\begin{lemma} \label{lma:cContinuous}
    For each boid $i\in\mathcal{A}$, if the functions $\mathbf{c}_i(t)$ or $\dot{\mathbf{c}}_i(t)$ are discontinuous at a time $t_1\in\mathbb{R}_{\geq 0}$, then the neighborhood $\mathcal{N}_i(t)$ must switch at $t=t_1$.
\end{lemma}
	
\begin{proof}
    Let $\mathcal{N}_i(t)$ be constant over some interval $[t_1, t_2]\subset\mathbb{R}_{\geq0}$.
    By Definition \ref{def:neighborhood}, $\mathbf{c}_i(t)$ is a sum of $k\in\mathbb{N}$ continuous functions, hence $\mathbf{c}_i(t)$ is continuous for all $t\in[t_1, t_2]$.
    Thus, $\mathbf{c}_i(t)$ may only be discontinuous at instants when $\mathcal{N}_i(t)$ is switching.
    The proof for $\dot{\mathbf{c}}_i(t)$ is identical, and thus it is omitted.
\end{proof}

\begin{definition} \label{def:symmetry}
    For each boid $i\in\mathcal{A}$, let its neighborhood $\mathcal{N}_i(t)$ switch at some time $t_1\in\mathbb{R}_{\geq0}$.
    Let $\mathcal{O} = \mathcal{N}_i(t_1^-)\setminus\mathcal{N}_i(t_1^+)$ and $\mathcal{P} = \mathcal{N}_i(t_1^+)\setminus \mathcal{N}_i(t_1^-)$ be the boids which were removed and added to the set $\mathcal{N}_i(t)$ at time $t_1$, respectively.
    Then, we define a switch as \emph{symmetric} if it satisfies
    $\sum_{j\in\mathcal{O}} \mathbf{p}_j(t_1) = \sum_{j\in\mathcal{P}} \mathbf{p}_j(t_1)$ and $ \sum_{j\in\mathcal{O}} \mathbf{v}_j(t_1) = \sum_{j\in\mathcal{P}} \mathbf{v}_j(t_1).$
\end{definition}
Note that, due to the safety constraint \eqref{eq:safety}, it is only possible to satisfy Definition \ref{def:symmetry} if multiple boids are added to and removed from the neighborhood simultaneously.

\begin{lemma}\label{lma:discontinuous}
    For each boid $i\in\mathcal{A}$, the functions $\mathbf{c}_i(t)$ or $\dot{\mathbf{c}}_i(t)$ are discontinuous at some time $t_1\in\mathbb{R}_{\geq0}$ if and only if $\mathcal{N}_i(t)$ switches at $t=t_1$ and the switch is non-symmetric (Definition \ref{def:symmetry}).
\end{lemma}

\begin{proof}
    First we prove sufficiency. By Lemma \ref{lma:cContinuous}, a discontinuity in $\mathbf{c}_i(t)$ or $\dot{\mathbf{c}}_i(t)$ at $t_1$ implies that $\mathcal{N}_i(t)$ switches at $t_1$.
    
    We prove necessity by contradiction. Let $\mathbf{c}_i(t)$ be continuous and assume $\mathcal{N}_i(t)$ switches non-symmetrically at some time $t_1$. This implies that $\mathbf{c}_i(t_1^-) = \mathbf{c}_i(t_1^+),$  where the superscripts $^+$ and $^-$ correspond to the left and right limits of $t_1$, respectively.
    This implies $\sum_{j\in\mathcal{N}_i(t_1^-)} \mathbf{p}_j(t_1) = \sum_{j\in\mathcal{N}_i(t_1^+)} \mathbf{p}_j(t_1),$ since $\mathbf{p}_l(t)$ is continuous for all $l\in\mathcal{A}$.  
    We may remove the common elements in the last equatoin, which implies
    $ \sum_{j\in\mathcal{O}} \mathbf{p}_j(t_1) = \sum_{j\in\mathcal{P}} \mathbf{p}_j(t_1).$
    The same analysis holds for continuity of $\dot{\mathbf{c}}_i(t)$. This contradicts the hypothesis that the switching is non-symmetric.
\end{proof}

Thus, by Lemma \ref{lma:discontinuous}, for each boid $i\in\mathcal{A}$, a switch in $\mathcal{N}_i(t)$ implies that either $\mathbf{c}_i(t)$ or $\dot{\mathbf{c}}_i(t)$ is discontinuous unless the conditions in Definition \ref{def:symmetry} are satisfied. As Definition \ref{def:symmetry} relies on ideal symmetry conditions, we may assume that a switch will never be symmetric in a real system. Thus, for simplicity, we consider a switch in a boid's neighborhood to always lead to a discontinuity in $\mathbf{c}_i(t)$ or its derivative.

Finally, we present Property \ref{prp:discontinuous} of the optimal trajectory, which describes the impact of each boid's neighborhood on the task constraint.

\begin{property}\label{prp:discontinuous}
    For each boid $i\in\mathcal{A}$ traveling along the task-constrained arc, i.e., following a trajectory that exactly satisfies \eqref{eq:flocking}, if $\mathcal{N}_i(t)$ switches at some time $t_1\in\mathbb{R}_{\geq 0}$, and the switch is non-symmetric (Definition \ref{def:symmetry}), then boid $i$ must exit the constrained arc.
\end{property}
	
\begin{proof}
    We prove Property \ref{prp:discontinuous} by contrapositive. 
	Let $i\in\mathcal{A}$ be any boid in the flock traveling along the task-constrained arc over the interval $[t_1, t_2]\subset\mathbb{R}_{\geq0}$.
	This implies that the optimal trajectory of boid $i$, denoted $\mathbf{p}_i^*(t)$, must satisfy
    \begin{equation} \label{eq:tangency1}
        ||\mathbf{p}_i^*(t) - \mathbf{c}_i(t)||^2 - D^2 = 0, \quad \forall t\in[t_1, t_2].
    \end{equation}
    The optimal trajectory $\mathbf{p}_i^*(t)$ must be continuous. This implies $\mathbf{c}_i(t)$ must be continuous.
    Additionally, the derivative of \eqref{eq:tangency1} must also hold when the task constraint is active, i.e.,
    \begin{equation} \label{eq:tangency2}
        \Big(\mathbf{p}_i^*(t) - \mathbf{c}_i(t)\Big) \cdot \Big(\mathbf{v}_i^*(t) - \dot{\mathbf{c}}_i(t)\Big) = 0, \quad \forall t\in[t_1, t_2].
    \end{equation}
	Therefore, $\dot{\mathbf{c}}_i(t)$ must be continuous.
	This implies that if $\mathbf{c}_i(t)$ or $\dot{\mathbf{c}}_i(t)$ is discontinuous then boid $i$ can not be on the task-constrained arc.
	Thus, by Lemma \ref{lma:discontinuous}, boid $i$ can not be traveling along the constrained arc when $\mathcal{N}_i(t)$ switches, unless the switch is symmetric.
\end{proof}
	
By Property \ref{prp:discontinuous}, we may infer that in any physical system boid $i$ will exit any task-constrained arc whenever $\mathcal{N}_i(t)$ switches at some time $t_1$.

    By Definition \ref{def:neighborhood}, boid $i$ only knows the state information of its $k$-nearest neighbors.
    This information in insufficient to calculate the time of a future trajectory change, $t_1$, or the new neighborhood center state, $\mathbf{c}_i(t_1)$ and $\dot{\mathbf{c}}_i(t_1)$.
    Therefore we must consider the 
    case where $g_i(t_1) > 0$ or $\dot{g}_i(t_1) > 0$ in general.
    This motivates the inclusion of the slack variable $\eta_i^2(t)$ in Problem \ref{prb:flocking}. 

    The behavior induced by this model is non-smooth and nonlinear, and a full analysis of the imposed flock is beyond the scope of this paper. In general, it is necessary to rule out chattering and Zeno behavior at instants when the neighborhood switches.
    Additionally, for each boid $i\in\mathcal{A}$, an analysis of how $\eta_i(t)$ evolves with the system will be necessary to ensure that the flock remains cohesive. In general, it is necessary for $\eta_i(t)$ to be driven to zero in finite time. Otherwise boid $i$ may violate the task constraint indefinitely, leading to flock fragmentation.
    As a step toward analyzing the full system, we analyze the case where the communication topology is fixed and connected in the next section.

\subsection{The Fixed and Connected Topology Case}

    Next, we present several properties of Problem \ref{prb:flocking} for the case that the communication topology is fixed and connected. This ensures that for every $i\in\mathcal{A}$, $\mathbf{c}_i(t)$ is continuously differentiable everywhere by Lemma \ref{lma:discontinuous}. We may therefore also impose $\eta_i(t) = 0$ to make the task a hard constraint.
    In addition, we relax \eqref{eq:vConstraint} and \eqref{eq:uConstraint}, and instead, we only require $\mathbf{v}_i(t)$ and $\mathbf{u}_i(t)$ to be finite everywhere. Generally, imposing state and control limits does not add significant complexity to the problem \cite{Bryson1975AppliedControl}, and these two cases have been thoroughly explored in the literature \cite{Malikopoulos2018,Beaver2020AnAgents}.

First, we present Property \ref{prp:equilibrium}, which describes an optimal solution to Problem \ref{prb:flocking}.
  
\begin{property} \label{prp:equilibrium}
    If velocity consensus is achieved at a time $t_1\in\mathbb{R}_{\geq0}$, and $\eta_i^2(t_1) = 0$ for all $i\in\mathcal{A}$, then the globally optimal solution to Problem \ref{prb:flocking} is to maintain velocity consensus for all $i\in\mathcal{A}$ and for all $t\geq t_1$.
\end{property}
    
\begin{proof}
    Let every boid $i\in\mathcal{A}$ at some time $t_1\in\mathbb{R}_{\geq0}$ move with some consensus velocity $\mathbf{v}_c$ such that no constraint of Problem \ref{prb:flocking} is violated and  $\eta_i^2(t_1) = 0$.
    Next, let $i$ follow the trajectory $\mathbf{u}_j(t) = \mathbf{0}$ for all $t\geq t_1$.
    Then, $\mathbf{v}_i(t) = \mathbf{v}_j(t)$ for all $t\geq t_1$.
    Thus, the vector $\mathbf{s}_{ij}(t)$ is constant, and the safety constraint can never be violated for any $t\geq t_1$.
    This also implies that $\mathbf{v}_i(t) = \dot{\mathbf{c}}_i(t)$, and thus the vector $\mathbf{p}_i(t) - \mathbf{c}_i(t)$ is constant for all $t\geq t_1$.
    This implies that the task constraint can never be violated for any $t\geq t_1$, and also $\eta_i^2(t) = 0$ for all $t \geq t_1$.
\end{proof}

To derive an optimal control policy for Problem \ref{prb:flocking} we may apply Hamiltonian analysis \cite{Bryson1975AppliedControl}. First, we solve for the form of the optimal control policy for each set of constraints that may become active. As $\eta_i(t) = 0$, we have four possible arcs which boid $i\in\mathcal{A}$ may travel along: (1) none of the constraints become active -- $\mathbf{u}_i^*(t) = \mathbf{0}$; (2) boid $i$ moves unconstrained to an interior point -- $\mathbf{u}_i^*(t) = \mathbf{a} t + \mathbf{b}$ \cite{Malikopoulos2018}; (3) boid $i$ activates the safety constraint with some $j\in\mathcal{N}_i$, i.e., $\mathbf{u}_i^*(t) = \mathbf{u}_j^*(t)$ \cite{Beaver2020AnAgents}; and (4) boid $i$ activates the task constraint, i.e., $\mathbf{u}_i^*(t) = \ddot{\mathbf{c}}_i(t)$ \cite{Beaver2020AnAgents},
where $\mathbf{a}_i$ and $\mathbf{b}_i$ are constants of integration. We refer to the above cases as our \emph{optimal motion primitives}. The optimal control policy of boid $i$ is a piecewise function consisting of our four optimal motion primitives which are pieced together while satisfying optimality conditions.
Next, we present a result which characterizes the control input for an agent $i\in\mathcal{A}$, when $i$ activates the safety or task constraint.

\begin{theorem} \label{thm:uContinuous}
    Let a boid $i\in\mathcal{A}$ transition to a task, or safety-constrained arc at some time $t_1\in\mathbb{R}_{\geq0}$ for the fixed-topology case. Then the control input $\mathbf{u}_i(t_1)$ is continuous.
\end{theorem}
\begin{proof}
    Let boid $i\in\mathcal{A}$ transition to a task-constrained arc at some $t=t_1$.
    In this case, boid $i$ must satisfy the tangency conditions (see \cite{Bryson1975AppliedControl}, pp. 101)
    \begin{align} \label{eq:N}
        N_i(t_1, \mathbf{x}_i(t_1)) = \begin{bmatrix}
            \mathbf{s}_i(t_1)\cdot\mathbf{s}_i(t_1) - D^2 \\
            \dot{\mathbf{s}}_i(t_1)\cdot\mathbf{s}_i(t_1)
        \end{bmatrix} = \mathbf{0}, \\
        \ddot{\mathbf{s}}_i(t_1)\cdot\mathbf{s}_i(t_1) + \dot{\mathbf{s}}_i(t_1)\cdot\dot{\mathbf{s}}_i(t_1) = 0. \label{eq:constraint}
    \end{align}
    The vectors $\mathbf{s}_i(t)$ and $\dot{\mathbf{s}}_i(t)$ are functions of $\mathbf{c}_i(t)$ and $\dot{\mathbf{c}}_i(t)$, which are known functions of time. Therefore,
    \begin{equation} \label{eq:partialNpartialT}
        \frac{\partial N_i(t, \mathbf{x}_i(t))}{\partial t} = 
        \begin{bmatrix}
        -2\,\dot{\mathbf{c}}_i(t) \cdot\mathbf{s}_i(t) \\
        -\dot{\mathbf{c}}_i(t)\cdot\dot{\mathbf{s}}_i(t) - \ddot{\mathbf{c}}_i(t)\cdot\mathbf{s}_i(t)
        \end{bmatrix}.
    \end{equation}
    In addition, we have
    \begin{equation} \label{eq:partialNpartialX}
        \frac{\partial N_i(t, \mathbf{x}_i(t))}{\partial \mathbf{x}_i} = \begin{bmatrix}
        2\,\mathbf{s}_i^T(t), \mathbf{0} \\ \dot{\mathbf{s}}_i^T(t), \mathbf{s}_i^T(t)
        \end{bmatrix}.
    \end{equation}
    Finally, at time $t_1$, the optimality conditions are \cite{Bryson1975AppliedControl}
    \begin{align}
        \boldsymbol{\lambda}_i^T(t_1^-) = \boldsymbol{\lambda}_i^T(t_1^+) = \boldsymbol{\pi}\cdot \frac{\partial N_i(t, \mathbf{x}_i(t)}{\partial \mathbf{x}_i(t)}\Bigg|_{t_1},\label{eq:lambdaJump}\\
        \frac{1}{2}||\mathbf{u}_i(t_1^+)||^2 + \boldsymbol{\lambda}_i(t_1^+)\cdot\dot{\mathbf{x}}_i(t_1^+) - \frac{1}{2}||\mathbf{u}_i(t_1^-)||^2 \notag\\ 
        - \boldsymbol{\lambda}_i(t_1^-)\cdot\dot{\mathbf{x}}_i(t_1^-) = \boldsymbol{\pi}\cdot \frac{\partial N_i(t, \mathbf{x}_i(t)}{\partial t}\Bigg|_{t=t_1}, \label{eq:hJump}
    \end{align}
    where $\boldsymbol{\lambda}_i(t)$ is the state covector and $\boldsymbol{\pi}_i$ is a $2\times1$ constant vector.
    Note the constraint does not appear in the Hamiltonian \eqref{eq:hJump} since it becomes zero at $t_1^+$ through \eqref{eq:constraint}.
    Substituting \eqref{eq:partialNpartialX}-\eqref{eq:lambdaJump} into \eqref{eq:hJump} and simplifying yields
    \begin{equation}
        ||\mathbf{u}_i(t_1^+)||^2 + ||\mathbf{u}_i(t_1^-)||^2 - 2\,\mathbf{u}_i(t_1^-)\cdot\mathbf{u}_i(t_1^+) = 0,
    \end{equation}
    which has the real solution $\mathbf{u}_i(t_1^-) = \mathbf{u}_i(t_1^+)$. Thus $\mathbf{u}_i(t_1)$ is continuous.
    \\
    The proof for the safety constraint is identical and thus we omitted it.
\end{proof}

\begin{corollary} \label{cor:allContinuous}
The optimal control input of each boid $i\in\mathcal{A}$ is continuous everywhere for the fixed topology case.
\end{corollary}

\begin{proof}
    By Theorem \ref{thm:uContinuous}, the control input $\mathbf{u}_i$ is continuous when boid $i$ enters a task or safety-constrained arc. When $i$ enters an unconstrained arc the covectors $\boldsymbol{\lambda}_i(t)$ and Hamiltonian are continuous, and the vector $N_i(t) = 0$ in \eqref{eq:N}. From the optimality conditions and Euler-Lagrange equations it is straightforward to show that the control input is continuous.
\end{proof}

Finally we present a proof of convergence to velocity consensus for the fixed and connected topology case trough Lemma \ref{lma:taskActive} and Theorem \ref{thm:convergence}.

\begin{lemma} \label{lma:taskActive}
    If a boid $i\in\mathcal{A}$ satisfies $||\mathbf{v}_i(t)|| > ||\dot{\mathbf{c}}_i(t)||$ while traveling along an unconstrained arc, then there exists some $t_1\in\mathbb{R}_{\geq0}$ such that the task constraint becomes active at $t_1$.
\end{lemma}

\begin{proof}
    Let $\mathbf{s}_{i}(t) = \mathbf{p}_i(t) - \mathbf{c}_i(t)$.
    By the triangle inequality $||\mathbf{\dot{s}}_i(t)|| > 0$.
    The speed profile imposed by all of our motion primitives is a polynomial \cite{Beaver2020AnAgents}, therefore $\mathbf{c}_i(t)$ cannot asymptotically approach $\mathbf{v}_i(t)$.
    Thus, there exists some finite $t_1\in\mathbb{R}_{\geq0}$ such that $||\mathbf{s}_i(t_1)|| = D$.
\end{proof}

\begin{theorem} \label{thm:convergence}
    For each boid $i\in\mathcal{A}$, if
    the parameter $D$ is large enough that the boids do not always follow safety-constrained trajectories, then there exists some time $t_1$ such that all boids $j\in\mathcal{N}_i$ achieve velocity consensus.
\end{theorem}

\begin{proof}
    First, we consider the case where the boids never activate the safety constraint. 
    %
    Let $M(t) = \{i\in\mathcal{A} ~:~ ||\mathbf{v}_i(t)|| > ||\dot{\mathbf{c}}_i(t)   \}$.
    %
    By Lemma \ref{lma:taskActive} there exists some finite time $t_1\in\mathbb{R}_{\geq0}$ such that boid $m\in M(t)$ activates its task constraint.
    Then, until $||\mathbf{v}_i(t)|| \leq ||\dot{\mathbf{c}}_i(t)||$, we may generate a sequence $\{t_n\}, n\in\mathbb{N}$, such that $||\mathbf{v}_{m(t_n)}(t_{n})|| > ||\dot{\mathbf{c}}_{m(t_{n+1})}(t_{n+1}) || = ||\mathbf{v}_{m(t_n)}(t_{n+1})||$ for all $m\in M(t)$.
    Let $l(t) = \arg\min_{i\in\mathcal{A}} ||\mathbf{v}_i(t)||$.
    Following the same procedure, there exists a sequence $\{t_n\}, n\in\mathbb{N}$, such that $||\mathbf{v}_l(t_{n})|| < ||\mathbf{v}_m(t_{n+1})||$ until $||\mathbf{v}_i(t)|| \geq ||\dot{\mathbf{c}}_i(t)||$. Thus, each boid $i\in\mathcal{A}$ must satisfy $||\mathbf{v}_i(t)|| = ||\dot{\mathbf{c}}_i(t)||$ asymptotically. 
    
    Next, let $||\mathbf{v}_i(t)|| = ||\dot{\mathbf{c}}_i(t)||$ for all $i\in\mathcal{A}$. Select a boid $j$ which satisfies $j = \arg\min_{j\in\mathcal{A}}||\mathbf{v}_j(t)||$. As $||\mathbf{v}_j(t)|| = ||\dot{\mathbf{c}}_j(t)||$ and boid $j$ has the minimum speed in $\mathcal{A}$, it must be true that $||\mathbf{v}_j|| = ||\mathbf{v}_k||$ for all $k\in\mathcal{N}_j$. As the agent topology is connected we may recursively apply this reasoning to find $||\mathbf{v}_i(t)|| = ||\mathbf{v}_j(t)||$ for all $i,j\in\mathcal{A}$.
    Following similar logic it can be shown that $\mathbf{v}_i(t) = \dot{\mathbf{c}}_i(t)$, and thus velocity consensus is achieved asymptotically.

    Finally, let $t^0$ be the time that boid $i\in\mathcal{A}$ plans its trajectory, and allow the safety constraint to become active at some time $t_1 > t^0$ for another boid $j\in\mathcal{N}_i(t^0)$.
    By definition of the optimal motion primitives, $\mathbf{u}_j(t) = \mathbf{u}_i(t)$ for the duration that the constraint is active. Thus, either agent $i$ eventually exits the safety-constrained arc and Lemma \ref{lma:taskActive} holds, or $\mathbf{v}_i(t) = \mathbf{v}_j(t)$. If $||\mathbf{v}_i(t)|| \neq ||\dot{\mathbf{c}}_i(t)||$ for $t > t_1$, then Lemma \ref{lma:taskActive} holds.
\end{proof}

Theorem \ref{thm:convergence} guarantees that any flock with a fixed topology will achieve velocity consensus. By the design of the task constraint \eqref{eq:constraint}, any two boids $i,j\in\mathcal{A}$ will be contained within a ball of diameter $N\cdot D$ centered on the flock. 

Supplementary information and simulation results of the proposed flocking controller can be found at: https://sites.google.com/view/ud-ids-lab/cdflock.

\section{Concluding Remarks} \label{sec:conclusion}

In this paper, we proposed a set of flocking rules under the constraint-driven paradigm for multi-robot systems.
We translated these rules into an optimal control problem and gave several properties of the optimal solution.
In addition, we motivated the inclusion of a time-varying slack variable in the formulation and discussed the challenges of planning trajectories in multi-agent problems with partial state observation.
We listed the set of optimal control motion primitives and proved that the optimal control policy is a continuous function. We also showed that the flock will achieve velocity consensus under a fixed topology.

A direction of future research include the extension of Theorem \ref{thm:convergence} to cover the dynamic topology case. It is likely that a methodology similar to \cite{Tanner2007} can be used to prove that velocity consensus is achieved despite the discontinuous behavior that appears in the dynamic case.
Sufficient conditions for the slack variable to guarantee flocking behavior is an area of ongoing research, as well as analyzing the trade-off between energy consumption and flock cohesion by the parameter $\alpha_i$. Finally, the inclusion of environmental obstacles into Problem \ref{prb:flocking} and an additional task constraint to influence the flock's motion require further investigation.

\bibliographystyle{IEEEtran}
\bibliography{mendeley, IDS_Pubs, BibliographyFull}

\begin{thebibliography}{10}
\providecommand{\url}[1]{#1}
\csname url@rmstyle\endcsname
\providecommand{\newblock}{\relax}
\providecommand{\bibinfo}[2]{#2}
\providecommand\BIBentrySTDinterwordspacing{\spaceskip=0pt\relax}
\providecommand\BIBentryALTinterwordstretchfactor{4}
\providecommand\BIBentryALTinterwordspacing{\spaceskip=\fontdimen2\font plus
\BIBentryALTinterwordstretchfactor\fontdimen3\font minus
  \fontdimen4\font\relax}
\providecommand\BIBforeignlanguage[2]{{%
\expandafter\ifx\csname l@#1\endcsname\relax
\typeout{** WARNING: IEEEtran.bst: No hyphenation pattern has been}%
\typeout{** loaded for the language `#1'. Using the pattern for}%
\typeout{** the default language instead.}%
\else
\language=\csname l@#1\endcsname
\fi
#2}}

\bibitem{Malikopoulos2018}
A.~A. Malikopoulos, C.~G. Cassandras, and Y.~J. Zhang, ``{A decentralized
  energy-optimal control framework for connected automated vehicles at
  signal-free intersections},'' \emph{Automatica}, vol.~93, no. April, pp.
  244--256, 2018.

\bibitem{Lindsey2012ConstructionTeams}
Q.~Lindsey, D.~Mellinger, and V.~Kumar, ``{Construction with quadrotor
  teams},'' \emph{Autonomous Robots}, 2012.

\bibitem{Corts2009}
J.~Cortes, ``{Global formation-shape stabilization of relative sensing
  networks},'' in \emph{Proceedings of the American Control Conference}, 2009.

\bibitem{Reynolds1987}
C.~W. Reynolds, ``{Flocks, herds and schools: A distributed behavioral
  model},'' \emph{Computer Graphics}, vol.~21, no.~4, pp. 25--34, 1987.

\bibitem{Olfati-Saber2006FlockingTheory}
R.~Olfati-Saber, ``{Flocking for multi-agent dynamic systems: Algorithms and
  theory},'' \emph{IEEE Transactions on Automatic Control}, vol.~51, no.~3, pp.
  401--420, 3 2006.

\bibitem{Egerstedt2018RobotAutonomy}
M.~Egerstedt, J.~N. Pauli, G.~Notomista, and S.~Hutchinson, ``{Robot ecology:
  Constraint-based control design for long duration autonomy},'' pp. 1--7, 1
  2018.

\bibitem{Bajec2009OrganizedBirds}
I.~L. Bajec and F.~H. Heppner, ``{Organized flight in birds},'' \emph{Animal
  Behaviour}, vol.~78, no.~4, pp. 777--789, 10 2009.

\bibitem{Thiollay1998FlockingHypothesis}
J.-M. Thiollay and M.~Jullien, ``{Flocking behaviour of foraging birds in a
  neotropical rain forest and the antipredator defence hypothesis},''
  \emph{IBIS}, vol. 140, pp. 382--394, 1998.

\bibitem{Malikopoulos2016b}
A.~A. Malikopoulos, ``A duality framework for stochastic optimal control of
  complex systems,'' \emph{IEEE Transactions on Automatic Control}, vol.~61,
  no.~10, pp. 2756--2765, 2016.

\bibitem{Beaver2020AnFlocking}
L.~E. Beaver, C.~Kroninger, and A.~A. Malikopoulos, ``{An Optimal Control
  Approach to Flocking},'' in \emph{Proceedings of the 2020 American Control
  Conference (to appear)}, 2020.

\bibitem{Ibuki2020Optimization-BasedBodies}
T.~Ibuki, S.~Wilson, J.~Yamauchi, M.~Fujita, and M.~Egerstedt,
  ``{Optimization-Based Distributed Flocking Control for Multiple Rigid
  Bodies},'' \emph{IEEE Robotics and Automation Letters}, vol.~5, no.~2, pp.
  1891--1898, 4 2020.

\bibitem{Notomista2019Constraint-DrivenSystems}
G.~Notomista and M.~Egerstedt, ``{Constraint-Driven Coordinated Control of
  Multi-Robot Systems},'' in \emph{Proceedings of the 2019 American Control
  Conference}, 2019.

\bibitem{Lindemann2019ControlTasks}
L.~Lindemann and D.~V. Dimarogonas, ``{Control barrier functions for signal
  temporal logic tasks},'' \emph{IEEE Control Systems Letters}, vol.~3, no.~1,
  pp. 96--101, 1 2019.

\bibitem{Tanner2007}
H.~G. Tanner, A.~Jadbabaie, and G.~J. Pappas, ``{Flocking in fixed and
  switching networks},'' \emph{IEEE Transactions on Automatic Control},
  vol.~52, no.~5, pp. 863--868, 2007.

\bibitem{Koren1991PotentialNavigation}
Y.~Koren and J.~Borenstein, ``{Potential Field Methods and their Inherent
  Limitations for Mobile Robot Navigation},'' in \emph{Proceedings of the 1991
  IEEE International Conference on Robotics and Automation}, 1991.

\bibitem{Vasarhelyi2018OptimizedEnvironments}
G.~V{\'{a}}s{\'{a}}rhelyi, C.~Vir{\'{a}}gh, G.~Somorjai, T.~Nepusz, A.~E.
  Eiben, and T.~Vicsek, ``{Optimized flocking of autonomous drones in confined
  environments},'' \emph{Science Robotics}, vol.~3, no.~20, 2018.

\bibitem{Ballerini2008InteractionStudy}
M.~Ballerini, N.~Cabibbo, R.~Candelier, A.~Cavagna, E.~Cisbani, I.~Giardina,
  V.~Lecomte, A.~Orlandi, G.~Parisi, A.~Procaccini, M.~Viale, and
  V.~Zdravkovic, ``{Interaction ruling animal collective behavior depends on
  topological rather than metric distance: Evidence from a field study},''
  \emph{Proceedings of the National Academy of Sciences of the United States of
  America}, vol. 105, no.~4, pp. 1232--1237, 2008.

\bibitem{Cristiani2011EffectsGroups}
E.~Cristiani, P.~Frasca, and B.~Piccoli, ``{Effects of anisotropic interactions
  on the structure of animal groups},'' \emph{Journal of Mathematical Biology},
  vol.~62, no.~4, pp. 569--588, 4 2011.

\bibitem{Rezaee2014}
H.~Rezaee and F.~Abdollahi, ``{A decentralized cooperative control scheme with
  obstacle avoidance for a team of mobile robots},'' \emph{IEEE Transactions on
  Industrial Electronics}, 2014.

\bibitem{VanDenBerg2011ReciprocalObstacles}
J.~Van Den~Berg, J.~Snape, S.~J. Guy, and D.~Manocha, ``{Reciprocal collision
  avoidance with acceleration-velocity obstacles},'' in \emph{Proceedings -
  IEEE International Conference on Robotics and Automation}, 2011, pp.
  3475--3482.

\bibitem{Hu2018Self-triggeredSystems}
Y.~Hu, J.~Zhan, and X.~Li, ``{Self-triggered distributed model predictive
  control for flocking of multi-agent systems},'' \emph{IET Control Theory {\&}
  Applications}, vol.~12, no.~18, pp. 2441--2448, 12 2018.

\bibitem{Dave2019a}
A.~Dave and A.~A. Malikopoulos, ``{Decentralized Stochastic Control in
  Partially Nested Information Structures},'' in \emph{IFAC-PapersOnLine},
  Chicago, IL, USA, 2019.

\bibitem{Morgan2016}
D.~Morgan, G.~P. Subramanian, S.-J. Chung, and F.~Y. Hadaegh, ``{Swarm
  assignment and trajectory optimization using variable-swarm, distributed
  auction assignment and sequential convex programming},'' \emph{International
  Journal of Robotics Research}, vol.~35, no.~10, pp. 1261--1285, 2016.

\bibitem{Beaver2019AGeneration}
L.~E. {Beaver} and A.~A. {Malikopoulos}, ``A decentralized control framework
  for energy-optimal goal assignment and trajectory generation,'' in \emph{2019
  IEEE 58th Conference on Decision and Control (CDC)}, 2019, pp. 879--884.

\bibitem{Turpin}
M.~Turpin, K.~Mohta, N.~Michael, and V.~Kumar, ``{Goal Assignment and
  Trajectory Planning for Large Teams of Aerial Robots},'' \emph{Proceedings of
  Robotics: Science and Systems}, vol.~37, pp. 401--415, 2013.

\bibitem{Beaver2020DemonstrationCity}
L.~E. Beaver, B.~Chalaki, A.~M.~I. Mahbub, L.~Zhao, R.~Zayas, and A.~A.
  Malikopoulos, ``Demonstration of a time-efficient mobility system using a
  scaled smart city,'' \emph{Vehicle System Dynamics}, vol.~58, no.~5, pp.
  787--804, 2020.

\bibitem{chalaki2019optimal}
B.~Chalaki and A.~A. Malikopoulos, ``An optimal coordination framework for
  connected and automated vehicles in two interconnected intersections,'' in
  \emph{2019 IEEE Conference on Control Technology and Applications
  (CCTA)}.\hskip 1em plus 0.5em minus 0.4em\relax IEEE, 2019, pp. 888--893.

\bibitem{Lin2004a}
J.~Lin, A.~Morse, and B.~Anderson, ``{The multi-agent rendezvous problem - the
  asynchronous case},'' in \emph{Proceedings of the 43rd IEEE Conference on
  Decision and Control}, 2004.

\bibitem{Luis2019TrajectoryControl}
C.~E. Luis and A.~P. Schoellig, ``{Trajectory Generation for Multiagent
  Point-To-Point Transitions via Distributed Model Predictive Control},''
  \emph{IEEE Robotics and Automation Letters}, vol.~4, no.~2, pp. 375--382, 4
  2019.

\bibitem{Zhan2013FlockingMeasurements}
J.~Zhan and X.~Li, ``{Flocking of multi-agent systems via model predictive
  control based on position-only measurements},'' \emph{IEEE Transactions on
  Industrial Informatics}, vol.~9, no.~1, pp. 377--385, 2013.

\bibitem{Lyu2019MultivehicleControl}
Y.~Lyu, J.~Hu, B.~M. Chen, C.~Zhao, and Q.~Pan, ``{Multivehicle Flocking With
  Collision Avoidance via Distributed Model Predictive Control},'' \emph{IEEE
  Transactions on Cybernetics}, pp. 1--12, 10 2019.

\bibitem{Bryson1975AppliedControl}
A.~E.~J. Bryson and Y.-C. Ho, \emph{{Applied Optimal Control: Optimization,
  Estimation, and Control}}.\hskip 1em plus 0.5em minus 0.4em\relax John Wiley
  and Sons, 1975.

\bibitem{Beaver2020AnAgents}
L.~E. Beaver and A.~A. Malikopoulos, ``{An Energy-Optimal Framework for
  Assignment and Trajectory Generation in Teams of Autonomous Agents},''
  \emph{Systems \& Control Letters}, vol. 138, 2020.

\end{thebibliography}

\end{document}